\documentclass[twocolumn]{llncs}
\usepackage{microtype} 
\let\llncssubparagraph\subparagraph
\let\subparagraph\paragraph
\usepackage[compact]{titlesec} 
\let\subparagraph\llncssubparagraph

\usepackage{geometry}
\usepackage{graphicx,amssymb,amsmath}
\usepackage[skip=5pt]{caption} 
\usepackage{subcaption}
\usepackage{hyperref}
\usepackage{algorithm}
\usepackage{algorithmic}
\usepackage{comment}
\usepackage{orcidlink}
\usepackage{enumitem}          
\usepackage{multicol}
\usepackage{etoolbox}
\usepackage{todonotes}
\AfterEndEnvironment{proof}{\vspace{-0.6\baselineskip}}

\newlength{\gembrespace}
\newcommand{\compressenvironments}[1]{%
  \forcsvlist{\compressenvironment}{#1}%
}
\newcommand{\compressenvironment}[1]{%
  \BeforeBeginEnvironment{#1}{\vspace{\gembrespace}}%
  \AfterEndEnvironment{#1}{\vspace{\gembrespace}}%
}

\compressenvironments{theorem,lemma,obs,definition,corollary,proposition}
\newtheorem{obs}{Observation}

\title{$M$-Guarding Polygons In $K$-Visibility}

\author{
Yeganeh Bahoo\inst{1}
\and
Ahmad Kamaludeen\inst{1}
}
\institute{
Department of Computer Science, Toronto Metropolitan University, Toronto, Canada\\
\email{\{bahoo,ahmad.kamaludeen\}@torontomu.ca}
}

\begin{document}

\maketitle

\section{Introduction}\label{sec:intro}
The Art Gallery Problem (AGP) is a classic challenge of ensuring every point in a polygon $P$ is seen by at least one guard~\cite{KUMARGHOSH199175,urrutia2000art,orourke1987art}. 
We study a variation with two key complexities: \textit{$M$-guarding}, where each point must be seen by at least $M$ guards, and \textit{$k$-visibility}, where a line of sight can penetrate up to $k$ polygon edges. The concept of $k$-visibility originates from the ``Superman Problem''~\cite{supermanmouawad1994} and has since been adapted to AGP variants~\cite{ModemIO_Aichholzer2015}, with efficient algorithms developed for computing $k$-visibility regions~\cite{time-space_BAHOO2019,vis_region_Bahoo2020}. The notion of $M$-guarding has also been explored in $0$-visibility~\cite{Belleville1994KGuardingPO}. However, the combination of these two areas remains largely unexplored, with limited related work in super-guarding~\cite{super_dark_rays,Super_Stars}. We present a theorem establishing that any polygon with holes can be 2-guarded under $k$-visibility where $k \geq 2$, which expands existing results in 0-visibility. We provide an algorithm that $M$-guards a polygon using a convex decomposition of the polygon. We show that for any even $k \geq 2$ there exists a placement of guards such that every point in the polygon is visible to $k + 2$ guards. Motivated by modern applications of $k$-visibility in wireless communications and mapping~\cite{yang2025channel,kim2024structure}, this paper presents a formal algorithm for $M$-guarding polygons using edge-restricted guards.
\section{Definitions}\label{sec:defn}
This section introduces the key terms and concepts essential for understanding the problem of $k$-visibility guarding in polygons with holes. We use some basic terms that are well defined in the field of Computational Geometry, such as \textit{polygon},\textit{edges} \textit{reflex vertices}, \textit{simple polygon}, \textit{pocket} and \textit{holes}. See~\cite{orourke1987art,shermer1992recent} for formal definitions.

A \textit{diagonal} is a line segment between two vertices that lies entirely within $P$. For this paper, an \textit{obstacle} is a polygon edge. Two points are \textit{$k$-visible} if the line segment between them crosses at most $k$ obstacles. A \textit{guard} is a point within a polygon that has the ability to ``see'' or cover some points in the interior of the polygon. A \textit{$k$-visibility guard} can see through up to $k$ obstacles. All guards in this work are restricted to the interior side of polygon edges. 

\section{Results}\label{sec:results}
In this section, we state the contributions that this paper will provide. As a reminder in this work, our goal is to guarantee that each point in any polygon $P$ can be seen by at least $M$ guards when $k \geq 2$ and to establish a structured method for placing guards to ensure $M$-visibility. To achieve this, we present an algorithm that uses an \textit{optimal convex decomposition}, and then assigns guards to edges while keeping the required coverage. \textit{Optimal convex decomposition} refers to the division of a polygon $P$ into non-overlapping polygons such that each edge of the polygon $P$ is an edge of exactly one sub-polygon. In this paper, they are divided using diagonals, and all pieces are convex. A \textit{real edge} refers to an edge of the original polygon. Also, a \textit{dual graph} refers to a graph where each node represents a convex piece, and an edge exists between two nodes if their corresponding convex pieces share an edge. We begin by recognizing that for a given simple polygon $P$, an optimal convex decomposition can be constructed. \cite{greene1983decomposition}

\subsection{Preliminary Results}




\begin{obs}\label{lemma:2} 
A non-convex simple quadrilateral $P$ has at most one pocket.
\end{obs}
\vspace{-4mm}

\begin{theorem}\label{thm:holes}
Every polygon with holes can be 2-guarded under $k$-visibility for $k \geq 2$.
\end{theorem}
\vspace{-4mm}
\begin{proof}
Let $P$ be a polygon, and suppose we are using guards with $k$-visibility for some $k \geq 2$. It is known that any simple polygon (i.e., a polygon without holes) can be 2-guarded under 0-visibility \cite{Belleville1994KGuardingPO}. Since $k$-visibility generalizes 0-visibility, any 2-guarding solution for a simple polygon in 0-visibility also holds for $k \geq 2$. Consider modifying $P$ by introducing a hole, creating a new polygon $P'$. The original 2-guarding solution may no longer suffice under 0-visibility, as the hole may obstruct visibility. To compensate, we place guards on each edge of the hole. Let $H$ denote this set of guards, each of which has $k$-visibility with $k \geq 2$.

 Each guard in $H$ can see through up to $k \geq 2$ walls, allowing them to view visibility to regions that became hidden when the hole was introduced. For convex holes, the guards on the hole’s boundary can collectively see into the previously visible region multiple times, up to $|H|$ times. Non-convex holes introduce reflex vertices that can create pockets. However, by placing a guard on each edge of the hole, we ensure that every such reflex vertex is adjacent to at least one guard. Because each guard can see through at least two edges, the $k$-visibility condition allows these guards to recover visibility to regions that would otherwise remain hidden. Thus, the union of the original 2-guarding solution and the guards placed on the hole’s edges ensures that $P'$ remains 2-guarded under $k$-visibility. \qed
\end{proof}
\begin{lemma}\label{thm:blocking_lemma}
Given a monotone simple polygon $P$ and an optimal convex decomposition of it, connecting 2 adjacent convex pieces creates 1 obstacle composed of 2 edges for any single ray of vision, which would be entirely visible in 2-visibility.
\end{lemma}
\vspace{-4mm}
\begin{proof}
Adding each diagonal during the decomposition can eliminate at most two reflex vertices, and reflex vertices always form a pocket \cite{ghosh1988recognizing}. Therefore we may eliminate one or two pockets. In the case where we have two pockets, they do not overlap, and there is no singular ray of vision that is blocked twice when convex pieces are merged \cite{keil2000polygon}. Therefore, by doing the opposite process of merging the polygon and removing diagonals that partition the polygon, at each step, we do not invalidate the 2-visibility guards by blocking the existing guards multiple times in one move, as long as there was no guard on the removed diagonal. \qed
\end{proof}
To explain the Lemma~\ref{thm:blocking_lemma} in simple terms, when we connect 2 pieces in the optimal convex decomposition, there should be a pocket separating them, thus forming a concave polygon. This concave polygon would be entirely visible to any $2$-visibility guard. This is the key to proving that merging pairs of convex polygon, which we will do later, maintains visibility enough that every point will be seen by 4 different guards. 
\\
Now, consider when we remove a diagonal that partitions the polygon, there may be a guard on that edge. We can move that guard to be in a different convex piece. We claim there exists a point another piece that can guard the entirety of the first convex piece.
\vspace{-2mm}
\begin{lemma} \label{thm:boundary_guarding}
Let $A$ be a convex polygon. If a point $g$ outside of $A$ sees the entire boundary $\partial A$, then $g$ sees every point in $A$.
\end{lemma}
\vspace{-4mm}
\begin{proof} 
Omitted for space, see Appendix \ref{appendix:proofs}
\end{proof}
\vspace{-2mm}
\begin{algorithm}[h]
    \caption{Sweeping algorithm for critical vertices in a pocket}
    \label{alg:sweep}
    \begin{algorithmic}[1]
        \REQUIRE A reflex chain defined by vertices $v_1, v_2, \dots, v_n$ forming a pocket, and an edge $e$ on the boundary of region $B$ facing the pocket.
        \ENSURE A segment $S \subset e$ such that every point in $S$ sees the entire boundary of the pocket under $k$-visibility.
        \STATE Let $v_n$ be the vertex at the end of the reflex chain (where the sweep begins).
        \STATE Let $v_1$ be the vertex at the start of the reflex chain.
        \STATE Initialize an empty list $C$ to store the critical vertices encountered by the sweep.
        \STATE Initialize a ray $L$ originating at $v_n$, rotating counterclockwise along the boundary of $B$.
        \FOR{each reflex vertex $u_i$ in the chain from $v_n$ to $v_1$}
            \IF{the ray $L$ encounters $u_i$ before intersecting any other chain edge}
                \STATE Append $u_i$ to the list $C$.
            \ENDIF
        \ENDFOR
        \STATE Let $k$ be the number of critical vertices found in the pocket.
        \STATE Let $w$ be the $\lceil k/2 \rceil$-th vertex in $C$.
        \STATE Let $L_w$ be the ray from $v_n$ through $w$.
        \STATE Let $p$ be the intersection of $L_w$ and edge $e$.
        \STATE Let $x$ be the reflex vertex at the endpoint of $e$.
        \STATE Let $S$ be the subsegment of $e$ between $p$ and $x$.
        \STATE \textbf{return} $S$.
    \end{algorithmic}
\end{algorithm}
\begin{figure}
    \centering
    \includegraphics[height=4cm]{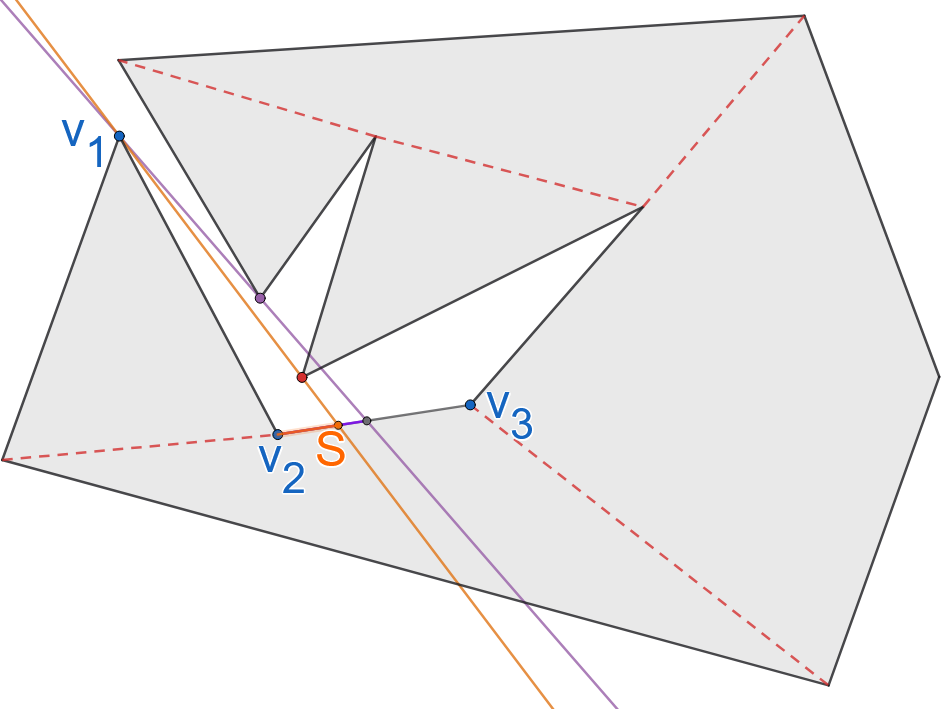}
    \caption{The ray is cast from $v_1$ along the reflex chain that forms the pocket. Once we find the ($k/2$)-th  critical reflex vertex (orange ray), the sweep ends}
    \label{fig:sweep}
\end{figure}

\begin{lemma} \label{thm:nonempty_S}
When the convex polygon $B$ in the decomposition is separated from another convex polygon $A$ by a diagonal $d$, there exists a non-empty set of points $S$ on the boundary $\partial P$ such that each point in $S$ sees the entire boundary of $A$ in 2-visibility.
\end{lemma}\vspace{-3mm}

\begin{proof}[Proof by Contradiction]
Let $d$ denote the diagonal that separates $A$ and $B$. Both $A$ and $B$ are convex. In our first case, $B$ is not interior. We consider the strong $k$-visibility polygon of each edge of $A$. The intersection of these visibility regions and $B$, is a set $T$ of points in $B$ that, if it exists, sees the entire boundary $A$ \cite{Segmentvis_bahoo2023}. Since both $A$ and $B$ are convex and share a diagonal $d$, a point $\epsilon$ away from $d$ in $B$ can see all of $A$. We aim to show that this intersection is non-empty and forms our desired set $S$. Assume for contradiction that $S$ is empty. That is, for all points $b \in B$, the point $b$ cannot see the entire boundary of the convex polygon $A$. Now, consider a point $b \in B$ that is $\epsilon$ distance away from the diagonal $d$. Since $A$ is convex, any obstruction $e$ preventing $b$ from seeing an edge of $A$ must lie on the boundary of $A$, in $B$, or on a third polygonal region $C$:
\begin{enumerate}[nosep, leftmargin=*]
    \item If $e$ is part of $A$, then visibility is blocked by the boundary of $A$, which is acceptable since the point $b$ acts as a 2-visibility guard, and therefore can see through this edge.
    \item $e$ lies in $B$: If $e$ is the edge that point $b$ lies on, it can see through this point by 2-visibility. If this is some other edge, then $B$ is no longer convex, and that is a contradiction.
    \item $e$ lies in a third polygonal region $C$: In this case, $e$ must lie between $b \in B$ and some portion of $\partial A$. But since $b$ is arbitrarily close to $d$, for $e$ to block the view, it must intersect the visibility cone from $b$ to $\partial A$ near $d$. However, due to convexity of both $A$ and $B$, and the fact that $b$ is infinitesimally close to $d$, such an edge $e$ must either intersect the interior of $A$, which contradicts the assumption that regions are convex and form a simple polygon, since the boundaries are intersecting. (Fig.~\ref{fig:hole_non_example}) 
\end{enumerate}
If $e$ lies outside of both $A$ and $B$, then $e$ cannot block any visibility from $b \in B$ to $A$ without intersecting $A$, again contradicting the assumption that $A$ is a valid convex polygon. Hence, the only possible obstructions to $b$ seeing all of $A$ are edges on $\partial A$. But since $b$ is arbitrarily close to $d$, and $A$ is convex, there exists a view corridor from $b$ to any edge of $A$. Therefore, our assumption must be false, and the set $S$ is non-empty. Next, our second case where we consider that $B$ can be an interior polygon. If $B$ is an interior polygon, then instead of $B$, we consider the convex piece $C$ that has a real edge and shares a vertex $v$ with $d$. All intervening convex pieces between $C$ and $A$ will be separated by diagonals originating from $v$. By the same argument above, we can see that there will be no edge $e$ blocking the visibility of the guards in $C$ from $A$ (other than the boundaries of $C$ and $A$ themselves), and so, the set $S$ is still non-empty as expected and our proof is complete. \qed
\end{proof}
Lemma \ref{thm:nonempty_S} implies that we have all the points between $\epsilon$ and the vertex to place the guard, and this region will be determined by any reflex vertices that are within the pocket.  (Fig.~\ref{fig:set_lemma})

\begin{figure}[h]
    \centering
    \includegraphics[height=3.5cm]{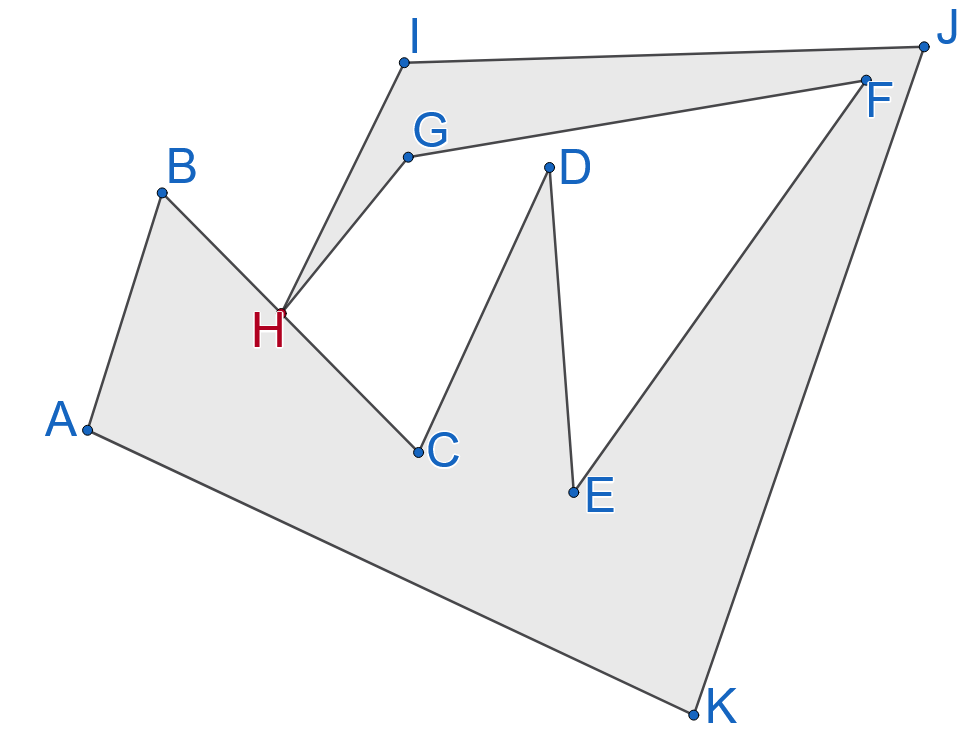}
    \caption{An invalid polygon with a hole. The vertex $H$ intersects the boundary of the outer polygon, violating the requirement that holes must lie entirely within the outer polygon}
    \label{fig:hole_non_example}
\end{figure}
\begin{figure}[h]
    \centering
    \includegraphics[height=3.5cm]{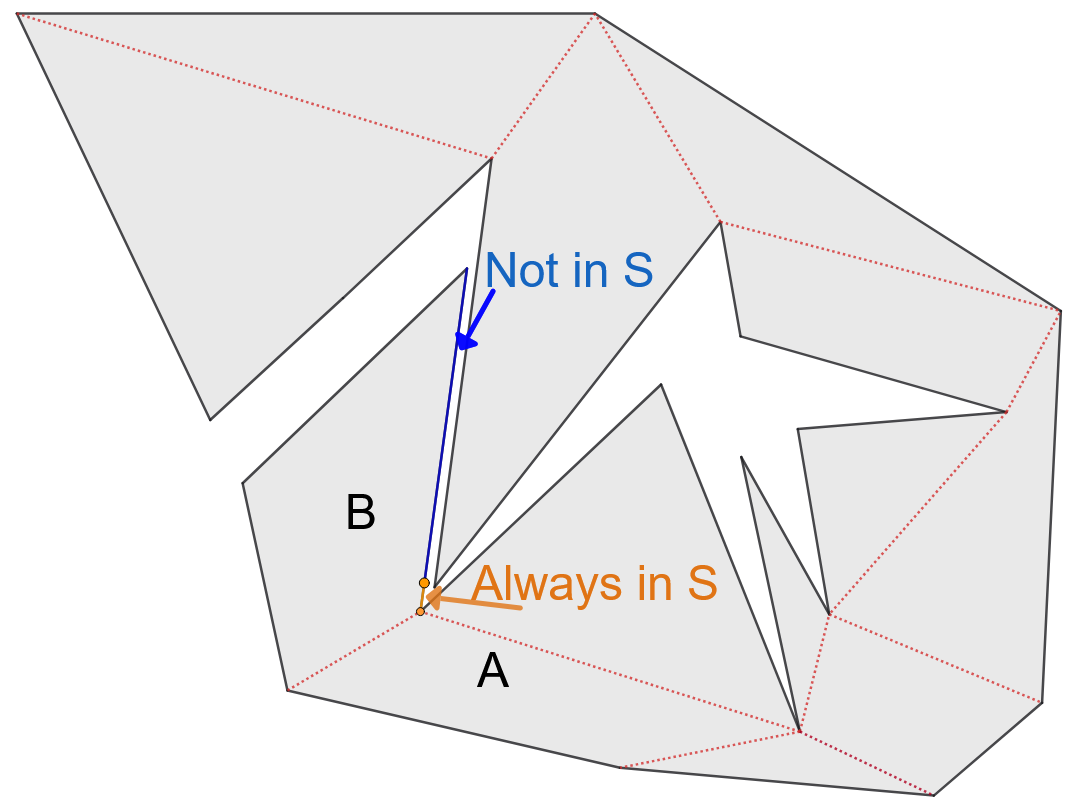}
    \caption{A polygon that demonstrates the size of the valid area for placing a guard in the adjacent piece.}
    \label{fig:set_lemma}
\end{figure}
\vspace{-4mm}
\begin{lemma}\label{thm:pushing_lemma}
Let $A$ and $B$ be adjacent convex polygons in a convex decomposition of a polygon, sharing an edge, which is the diagonal $d$. Suppose a 0-visibility guard is initially placed on $d$ to guard $A$. Then this guard may be repositioned to a point on a real edge that guards $A$ and $B$ in 2-visibility.
\end{lemma}
\vspace{-3mm}
\begin{proof}
By Lemma~\ref{thm:nonempty_S}, there exists a non-empty set of points $S$ on the boundary of the polygon $P$, arbitrarily close to $d$, such that each point in $S$ sees the entire boundary of $A$ under 2-visibility. By Lemma~\ref{thm:boundary_guarding}, any such point that sees all of $\partial A$ also sees the entire interior of $A$. Thus, a 2-visibility guard originally placed on the diagonal $d$ can be moved to a point in $S$ without loss of coverage, since there may only be 2 new edges blocking it. Since $S$ contains points that guard $\partial A$ in 2-visibility, placing the guard on one of these points maintains full 2-guarding of $A$.\qed
\end{proof}
We also claim that we do not adjust any single guard more than once. Since we are traversing the dual graph, we will be adjusting the guards forward into the next piece, or the one after that, until we reach the final ear. In this method, we will not backtrack, and thus we will not be adjusting 2 guards to the same edge.

\subsection{$(k+2)$-guarding with $k$-visibility} \label{sec:general k}
Given a polygon $P$, we propose Algorithm \ref{alg:pseudo_2} for $(k+2)$-guarding of $P$ using $k$-visibility. 
Before we go into detail, we provide a short overview. \\
We begin by decomposing the non-convex polygon $P$ into an optimal convex decomposition and constructing its dual graph. Adjacent convex regions are then paired along the dual graph wherever possible, leaving some regions unpaired (such as ears or leaves) when pairing is not feasible due to the structure of the dual. For each paired set of adjacent convex pieces, a guard is placed on a real edge or on the shared diagonal between them so that both regions are visible under $k$-visibility. For isolated convex regions, a guard is similarly placed on one of their real edges. Around each vertex, a small area containing only its two incident edges is considered. Within this area, every point can see at least one of the adjacent convex regions under $k$-visibility, ensuring local coverage of the boundary without obstruction by other parts of the polygon. Finally, each guard is relocated by a small $\varepsilon$ along the polygon boundary in the direction away from the root of the dual tree. This guarantees that the guards remain on real edges after merging adjacent convex pieces and that the entire polygon remains fully covered under $k$-visibility. See Fig.~\ref{fig:combined} for a visual example of the algorithm.
We have the following theorem:
\begin{theorem} \label{thm:3}
Given a polygon $P$ (possibly with holes) and a visibility parameter $k$, and let $k = 2i$ for some integer $i \geq 0$, there exists a placement of guards on the edges of $P$ such that every point in $P$ is $(k+2)$-guarded under $k$-visibility.
\end{theorem}\vspace{-3.5mm}
\begin{proof}
This can be accomplished following the steps of Algorithm \ref{alg:pseudo_2}. The proof of the correctness of this algorithm is as follows: We prove that every point in the polygon is visible to at least $(k+2)$ guards by induction on the convex pieces.
\\
\textbf{Base Case:} $n = 1$. For a single convex polygon, place one guard at the midpoint of each real edge of the ear. Since the polygon is convex, every point in the polygon is $k$-visible to each guard, and with at least $k+2$ edges, every point is seen by at least $k+2$ guards.
\\
\textbf{Inductive Step:} Assume that for some $n \geq 1$, every point in a polygon $P_n$ with $n$ convex pieces is $(k+2)$-guarded. We show this holds for a polygon with $n+1$ convex pieces. Consider a polygon with $n+1$ convex pieces. Removing one convex piece yields a polygon with $n$ convex pieces, where every point is $(k+2)$-guarded by assumption. When we add a convex piece to a polygon, it is always an ear, and therefore has 2 real edges as shown by the well-known Two Ears Theorem \cite{ears_meisters}. For the new convex piece, if there was a guard on edge that connects the $(n+1)$-th piece to the polygon, place one guard on the real edge within the intersection of the strong segment $k$-visibility polygons of the $n$-th piece, and place the other guard anywhere on an available free edge. Otherwise, we can place 2 guards onto the available real edges. Merging the piece into $P_n$ ensures that every point remains $(k+2)$-guarded, as long as we have placed the guards within the visibility polygons of all the segments of the adjacent piece, as proved in Lemma~\ref{thm:blocking_lemma}, Lemma~\ref{thm:nonempty_S}, and Lemma~\ref{thm:pushing_lemma}. It is also possible that there is a piece inside the pocket that blocks the visibility when we combine a non-monotone polygon with a convex piece (Fig~\ref{fig:set_lemma}), So we check if the next piece is completely visible to the segment. If it is not, then we need to use Lemma \ref{thm:nonempty_S}, and find the set of points $S$ that will guard the $X$ by finding the critical vertices between the two pieces, $X$ and $C$, which can be accomplished with Algorithm~\ref{alg:sweep}. A critical vertex is a reflex vertex whose adjacent edges lie in the same half-plane. When the sweep passes such a vertex, the visible region shrinks because the vertex acts as a spike that obstructs part of the pocket. Consequently, the feasible set $S$ on the segment is determined by the endpoint of the segment and these critical vertices. See Figure~\ref{fig:sweep} for a visual illustration. Therefore, every point in the polygon with $C+1$ convex pieces is visible to at least four guards, and so by the principle of mathematical induction, every point in a polygon with holes is visible to at least four guards. By induction, every point in a polygon with holes is $(k+2)$-guarded. \qed
\end{proof}
\begin{algorithm}[h!]
    \caption{$(k+2)$-Guarding a Polygon with $k$-visibility}
    \label{alg:pseudo_2}
    \begin{algorithmic}[1]
        \REQUIRE A non-convex polygon $P$ and a parameter $k$.
        \ENSURE Guard set $G$ that $(k+2)$-guards $P$ under $k$-visibility.
        \STATE Decompose $P$ into optimal convex pieces.
        \STATE Merge pieces to have $\ge k$ edges (except ears).
        \STATE Construct the dual graph of the convex decomposition.
        \STATE Let $C_0$ be the leftmost convex ear in the decomposition.
        \STATE Place one guard at the midpoint of each real edge of $C_0$.
        \FOR{each convex piece $C_i$ along the dual graph path starting from $C_0$}
            \STATE Perform the sweep-line algorithm (Algorithm~\ref{alg:sweep}) to find critical points in the pocket formed between $C_i$ and the next piece $C_{i+1}$.
            \IF{there are fewer than $k/2$ vertices in the pocket}
                \STATE Place guards at the midpoints of real edges of $C_i$.
            \ELSE
                \STATE Place the guards in the set found by Algorithm~\ref{alg:sweep}.
            \ENDIF
            \IF{guards cannot be placed on real edges because not enough edges are available}
                \STATE Place the guards on the diagonals.
            \ENDIF
            \STATE Merge $C_i$ with the current merged region.
            \IF{a guard is placed on a diagonal that is removed during merging}
                \STATE Relocate the guard $g_f$ to a point on a real edge of the adjacent region from which the entire boundary of the previous convex piece $C_{i-1}$ is visible under $k$-visibility, by computing the intersection of $C_{i-1}$ and the strong $k$-visibility polygons of every edge in $C_i$.
                \IF{there are no real edges in the adjacent region}
                    \STATE Let $e$ be the real edge connected to the vertex of the diagonal $d$.
                    \STATE Let $S$ be the set of points on $e$ found using Algorithm~\ref{alg:sweep} with respect to the pocket and edge $e$.
                    \STATE Relocate the guard to a point in $S$.
                \ENDIF
            \ENDIF
        \ENDFOR
        \RETURN{The set of $k$-visibility guards that $(k+2)$-guard the entire polygon.}
    \end{algorithmic}
\end{algorithm}
As a result of Algorithm~\ref{alg:pseudo_2} and Theorem~\ref{thm:3}, it follows that there exists an upper bound on the number of guards required to $k+2$-guard a concave polygon under $k$-visibility. 
\vspace{-4mm}
\begin{obs}
    For any concave polygon $P$ that is decomposed into $C$ convex pieces, at most $kC$ guards are sufficient to $(k+2)$-guard $P$.
\end{obs}\label{guard-bound}
\vspace{-1mm}
The decomposition used in our construction ensures that each convex piece can be $(k+2)$-guarded by at most $k$ guards placed on its boundary. Since the decomposition yields $C$ convex pieces, assigning at most $k$ guards per piece requires no more than $kC$ guards. Usually, as $k$ increases, we can view more of the polygon, thus our number of guards should decrease. However, in this problem, as $k$ increases, the maximum $M$-guarding possible increases up to the number of vertices $n$.
\vspace{0.5mm}

\begin{figure}[h!] 
    \centering 
    \begin{subfigure}[b]{0.48\columnwidth}
        \includegraphics[width=\textwidth]{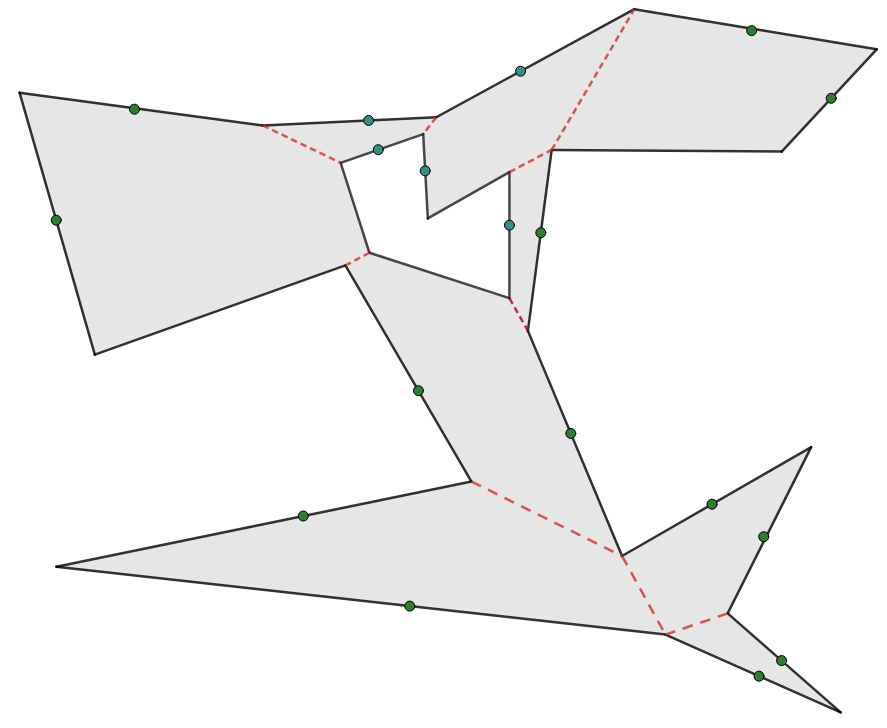}
        \caption{An example of a polygon after decomposition and 2-guarding}
        \label{fig:final}
    \end{subfigure}
    \hfill 
    \begin{subfigure}[b]{0.48\columnwidth}
        \includegraphics[width=\textwidth]{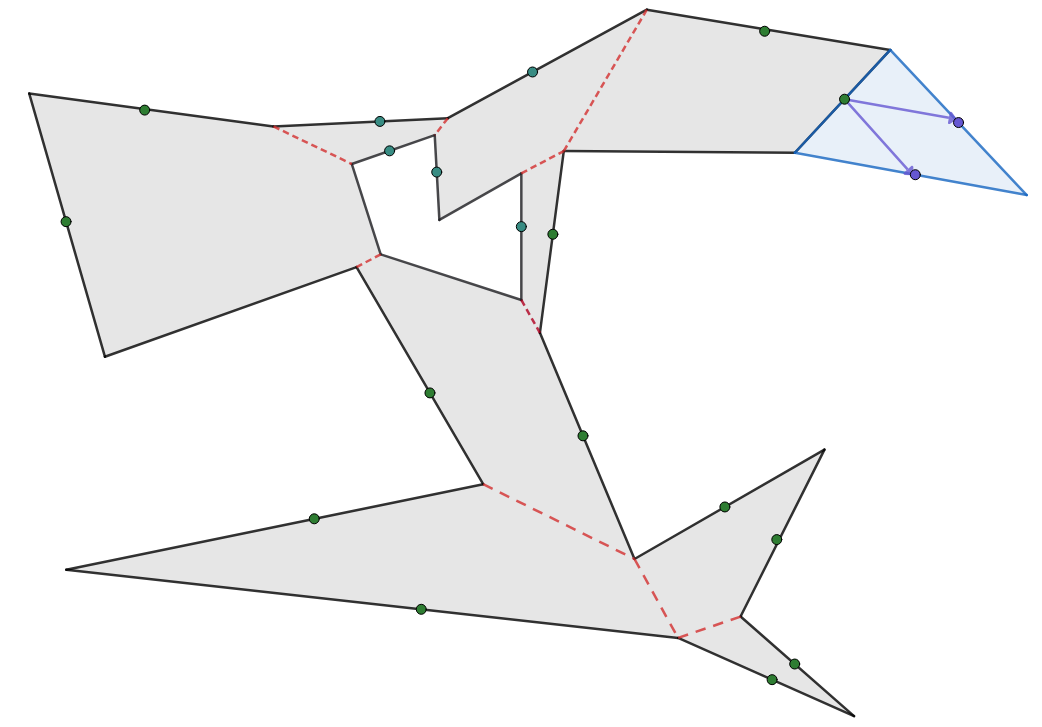}
        \caption{An example of a polygon with an additional piece.}
        \label{fig:induction}
    \end{subfigure}
    \caption{Panels (a) and (b) show different steps of Algorithm \ref{alg:pseudo_2}.}
    \label{fig:combined}
\end{figure}
\vspace{-2mm}
\section{Conclusion}\label{sec:future work} \vspace{-2mm}
In this work, we focused on addressing $M$-guarding under $k$-visibility without attempting to minimize the number of guards. Our approach allows for new strategies that ensure multiple guards cover an area under $k$-visibility via an efficient algorithm, and presents a bound for the number of guards facilitating practical applications where redundancy is more critical than minimizing number of guards deployed.
\\
This work was supported by the Natural Sciences and Engineering Research Council of Canada (NSERC).
\newpage
\small
\bibliography{references}

\appendix

\clearpage
\section{Proofs}\label{appendix:proofs}

\setcounter{lemma}{1}
\begin{lemma} \label{appendix:L4}
Let $A$ be a convex polygon. If a point $g$ outside of $A$ sees the entire boundary $\partial A$, then $g$ sees every point in $A$.
\end{lemma}

\begin{proof}
    Let $g$ be a guard located outside the convex polygon $A$, such that $g$ sees the entire boundary $\partial A$. Let $p$ be any point in the interior of $A$. Since $A$ is convex, any segment between two boundary points lies entirely within $A$. Because $g$ sees the entire boundary, we can select two boundary points $q_1$ and $q_2$ visible to $g$ such that the triangle $\triangle gq_1q_2$ contains $p$. Since $\overline{gq_1}$ and $\overline{gq_2}$ are unobstructed and $q_1, q_2 \in \partial A$, and since $A$ is convex, the triangle $\triangle gq_1q_2$ intersects $A$ only in its interior. Therefore, the line segment $\overline{gp}$ lies entirely within $\triangle gq_1q_2$ and hence within the union of $A$ and the guard’s view. So $p$ is visible to $g$. Hence, every point in $A$ is visible to the guard $g$, and the lemma holds.\qed
\end{proof}

\end{document}